\documentclass[conference]{IEEEtran}
\IEEEoverridecommandlockouts
\usepackage{cite}
\usepackage{amsmath,amssymb,amsfonts}
\usepackage{algorithmic}
\usepackage{graphicx}
\usepackage{textcomp}
\usepackage{xcolor}
\usepackage{comment}
\usepackage{subfigure}
\usepackage{array}
\usepackage{booktabs}

\usepackage[english]{babel}
\usepackage{amsthm}
\usepackage{amssymb}

\newtheorem{theorem}{Theorem}
\newtheorem{lemma}{Lemma}

\theoremstyle{plain}
\newtheorem{corollary}{Corollary}

\theoremstyle{plain}

\theoremstyle{remark}

\def\BibTeX{{\rm B\kern-.05em{\sc i\kern-.025em b}\kern-.08em
    T\kern-.1667em\lower.7ex\hbox{E}\kern-.125emX}}
\begin{document}

\title{Constellation Design in OFDM-ISAC over Data Payloads: From MSE Analysis to Experimentation \\

%% \thanks{Identify applicable funding agency here. If none, delete this.}
}

\author{\IEEEauthorblockN{Kawon Han}
\IEEEauthorblockA{
\textit{University College London}\\
London, United Kingdom \\
kawon.han@ucl.ac.uk}
\and
\IEEEauthorblockN{Kaitao Meng}
\IEEEauthorblockA{
\textit{The University of Manchester}\\
Manchester, United Kingdom \\
k.meng@manchester.ac.uk}
\and
\IEEEauthorblockN{Alexandra Chatzicharistou}
\IEEEauthorblockA{
\textit{University College London}\\
London, United Kingdom \\
a.chatzicharistou@ucl.ac.uk}
\and
\IEEEauthorblockN{Christos Masouros}
\IEEEauthorblockA{
\textit{University College London}\\
London, United Kingdom \\
c.masouros@ucl.ac.uk}
}

\maketitle

\begin{abstract}
Orthogonal frequency division multiplexing (OFDM) is one of the most widely adopted waveforms for integrated sensing and communication (ISAC) systems, owing to its high spectral efficiency and compatibility with modern communication standards. This paper investigates the sensing performance of OFDM-based ISAC for multi-target delay (range) estimation under specific radar receiver processing schemes. An estimation-theoretic framework is developed to characterize sensing performance with random communication payloads. We establish the fundamental limit of delay estimation accuracy by deriving the closed-form expression of the mean-square error (MSE) achieved using matched filtering (MF) and reciprocal filtering (RF) receivers. The results show that, in multi-target scenarios, the impact of signal constellations on the delay estimation MSE differs across receivers: MF performance depends on the fourth moment of the zero-mean, unit-power constellation in the presence of multiple targets, whereas RF performance depends on its inverse second moment, irrespective of the number of targets. Building on this analysis, we present a ISAC constellation design under specific receiver architecture that brings a receiver-dependent flexible trade-off between sensing and communication in OFDM-ISAC systems. The theoretical findings are validated through simulations and proof-of-concept experiments, and also the sensing and communication performance trade-off is experimentally shown with the proposed constellation design.
\end{abstract}


\section{Introduction}
Integrated sensing and communication (ISAC) has been attracting significant attention as a key enabler for next-generation wireless networks. Considerable research efforts have been devoted to realizing dual-functional radar–communication (DFRC) infrastructures, which allow sensing and communication functionalities to coexist within a unified system. In particular, achieving a flexible trade-off between sensing and communication performance is recognized as a critical requirement for ISAC systems, alongside the need for standard-compatible waveform designs \cite{luo2025isac}.

Among the candidate waveforms, orthogonal frequency division multiplexing (OFDM) remains a strong contender owing to its high spectral efficiency and widespread adoption in modern communication standards \cite{prasad2004ofdm}. Furthermore, recent studies have shown that cyclic-prefix OFDM (CP-OFDM) achieves the lowest range sidelobe levels for sensing, making it one of the most promising waveform choices for ISAC \cite{liu2025cp}. Building on these insights, OFDM-based ISAC systems enable efficient implementation while ensuring reliable joint sensing and communication performance.

Unlike conventional OFDM radar systems that employ optimized probing signals such as CAZAC sequences \cite{knill2021coded}, OFDM-based ISAC leverages communication data payloads, which are randomly drawn from a set of modulation constellation symbols. This alters the radar sensing performance, introducing a new trade-off between sensing and communication that depends on the modulation constellation. Recent works \cite{liu2025cp, liu2025uncovering} have presented analytical studies of the sidelobe levels of the ambiguity function with random ISAC signals, providing fundamental insights into how the modulation constellation impacts sensing performance. Also, \cite{geiger2025joint} provides the constellation-dependent detection probability under matched filtering (MF) in OFDM-ISAC. However, the ambiguity function characterizes radar performance only under MF, and thus does not directly capture performance under mismatched filtering techniques that mitigate sidelobes \cite{mcaulay1971optimal}.

Instead of the MF receiver, reciprocal filtering (RF) has been developed for OFDM-based radar sensing \cite{sturm2011waveform, wojaczek2018reciprocal}. RF suppresses target- and clutter-induced sidelobes by equalizing the data dependency at the receiver. However, this comes at the cost of a signal-to-noise ratio (SNR) loss compared to MF, which introduces a trade-off between sidelobe suppression and SNR. The work in \cite{keskin2025fundamental} provides a theoretical comparison of sidelobe levels between MF and RF outputs as a function of target SNR. Nevertheless, it focuses only on sidelobe levels and their impact on detection probability, without offering a rigorous analysis of multi-target range estimation performance. Therefore, a comprehensive estimation-theoretic study that characterizes the range estimation accuracy of OFDM-based ISAC under both MF and RF processing is still missing.

To address this gap, we develop a fundamental framework for OFDM-based ISAC using estimation-theoretic analysis under specific sensing receiver architectures. First, we present closed-form expressions for the mean-square error (MSE) of delay estimation with both MF and RF receivers, highlighting their dependence on the geometry of the modulation constellation and multi-target induced interference. Then, we formulate an ISAC constellation design problem for OFDM-ISAC systems with specific receiver processing, which brings new receiver-dependent trade-off between sensing and communication throughout geometric constellation shaping.The theoretical findings are validated through numerical simulations and proof-of-concept (PoC) experiments, thereby offering estimation-theoretic insights into OFDM-ISAC and demonstrating their practical relevance.

\section{System Model}
\subsection{Transmit Signal Model}
The ISAC transmit (TX) signal, utilizing \(N\) OFDM subcarriers, is modulated with random communication symbols drawn from the $M$-ary constellation set \(\mathcal{S} = \{s_1,s_2, ... ,s_M\}\). It is expressed as $\mathbf{x} = [x_0, x_1, \dots, x_{N-1}]$, where \( x_n \in \mathcal{S}, \forall n = 0,1, \dots, N-1 \). Without loss of generality, we assume that the constellation symbols are zero-mean and unit-variance as $\mathbb{E}\left[|x_n|^2\right] = 1, \forall x_n \in \mathcal{S}$, with statistical properties defined as follows:
\begin{align}
    \mathbb{E}\left[|x_n|^4\right] = \mu_{4}, \;\;\; 
    \mathbb{E}\left[|x_n|^{-2}\right]  = \nu_{-2}, 
    \quad \forall x_n \in \mathcal{S}, \label{eq1}
\end{align}
where $\mu_{4}$ denotes the fourth moment of the constellation corresponding to the kurtosis \cite{liu2025cp}, and $\nu_{-2}$ denotes the inverse second moment of the constellation \cite{han2025secure}. For a general $M$-ary PSK/QAM constellation, they are given by $\mu_{4} = \frac{1}{M}\sum_{m=1}^{M} |s_m|^4$ and $\nu_{-2} = \frac{1}{M}\sum_{m=1}^{M} |s_m|^{-2}$, respectively. This data-embedded signal is used for both communication and sensing. It is received by a communication user and is also reflected by targets, returning to the sensing receiver.

\subsection{Sensing System Model}
Suppose that $K$ target sources located at different ranges are sufficiently separated from each other, which is a common assumption in conventional radar systems. The frequency-domain received (RX) OFDM signal at the sensing receiver can be modeled as  
\begin{align}
    \mathbf{y}  = \mathbf{a}^T \mathbf{H} \mathbf{X} + \mathbf{z}, \label{RX_signal}
\end{align}
where $\mathbf{a} = [\alpha_{1}, \alpha_{2}, \dots, \alpha_{K}]^T \in \mathbb{C}^{K \times 1}$ denotes the complex amplitudes that incorporate the path loss and radar cross-section (RCS) of each target. The delay-channel matrix is expressed as $\mathbf{H} = [\mathbf{h}(\tau_1), \mathbf{h}(\tau_2), \dots , \mathbf{h}(\tau_K)]^T \in \mathbb{C}^{K \times N}$, where $\tau_{k}$ is the time-of-flight (TOF) from the ISAC transmitter to target $k$ and back to the receiver. The delay steering vector is defined as  
$\mathbf{h}(\tau) = \begin{bmatrix} 1, & e^{-j2 \pi \Delta f \tau}, & \dots, & e^{-j2 \pi (N-1)\Delta f \tau} \end{bmatrix}^T \in \mathbb{C}^{N \times 1},$ with subcarrier spacing $\Delta f = B/N$, where $B$ denotes the signal bandwidth. The transmitted signal $\mathbf{X}$ is the diagonal matrix of $\mathbf{x}$, i.e., $\mathbf{X} = \text{diag}(\mathbf{x})$. Finally, $\mathbf{z}$ denotes the additive white Gaussian noise (AWGN) at the sensing receiver, following $\mathbf{z} \sim \mathcal{CN}(0,\sigma^2 \mathbf{I}_N)$.

\section{Main Results}
In this section, we analyze the delay (or range) estimation performance of OFDM-ISAC systems under random communication signals via MSE of specific receivers; MF and RF.
\subsection{Delay Estimation MSE with Matched Filtering Receiver}
First, we derive the expected estimation error of the delay $\tau_k$ of source $k$ under MF receiver processing.
After matched filtering, i.e., multiplying $\mathbf{y}$ by $\mathbf{X}^H$, the output signal of MF $\mathbf{y}_{\text{MF}} = \mathbf{y}\mathbf{X}^H$ becomes
\begin{align}
    \mathbf{y}_{\text{MF}} = \mathbf{a}^T \mathbf{H} |\mathbf{X}|^2 + \mathbf{z}_{\text{MF}},
\end{align}
where $\mathbf{z}_{\text{MF}}$ follows the same noise characteristics as $\mathbf{z}$ due to assumption on unit-variance constellation. Let us formulate the general delay estimator that uses atoms in a set $\mathcal{A}=\{\mathbf{h}(\tau) \; | \; \tau \in \mathcal{T}\}$, where $\mathcal{T}$ denotes the delay candidate set with sufficient resolution. The inverse discrete Fourier transform (IDFT) is one efficient implementation of such estimator. This estimator can be defined as
\begin{align} \label{estimator_MF}
    \hat{\tau} = \arg \max_{\tau \in \mathcal{T}} \left| \mathbf{h}^H(\tau) \mathbf{y}_{\text{MF}}^T \right|.
\end{align}
Let us define $s_{\text{MF}}(\tau) = \mathbf{h}^H(\tau) \mathbf{y}_{\text{MF}}^T$, which is expressed as
\begin{equation}
    s_{\text{MF}}(\tau) = \sum_{k=1}^K \sum_{n=0}^{N-1} \alpha_k |x_n|^2 e^{j2\pi n \Delta f (\tau - \tau_k)} + \sum_{n=0}^{N-1} x_n^* z_n e^{j2\pi n \Delta f \tau},
\end{equation}
Based on this objective function \eqref{estimator_MF}, we define $f(\tau) = |s_{\text{MF}}(\tau)|^2$. Then, the following lemma approximates the delay estimation error in terms of $f(\tau)$.

\begin{lemma} \label{lemma1}
    Assuming that the source $k$ is the desired target and $\tau_k$ is its true delay, the biased error of the estimate $\tau_k$ is given by
    \begin{equation}
        \hat{\tau}_k = \tau_k - \frac{\dot{f}(\tau_k)}{\ddot{f}(\tau_k)}. \label{MSE1}
    \end{equation}
\end{lemma}
\renewcommand\qedsymbol{$\blacksquare$}
\begin{proof}
    Assuming that $f(\tau)$ has a sharp peak near $\tau_k$, a second-order Taylor expansion of $f(\tau)$ around $\tau_k$ yields
    \begin{align}
        f(\tau)  \approx f(\tau_k) + \dot{f}(\tau_k)(\tau - \tau_k) + \frac{1}{2} \ddot{f}(\tau_k)(\tau - \tau_k)^2. \label{Talyor}
    \end{align}
    The first derivative of \eqref{Talyor} in terms of $\tau$ is given by
    \begin{align}
        \dot{f}(\tau)  = \dot{f}(\tau_k) +  \ddot{f}(\tau_k)(\tau - \tau_k). \label{Talyor2}
    \end{align}
    Plugging the first-order optimality condition that $\dot{f}(\tau) \rightarrow 0$ as $\tau \rightarrow {\tau}_k$ into \eqref{Talyor2} yields \eqref{MSE1}, completing the proof.
\end{proof}
Now, we evaluate the MSE by taking the squared expectation over \eqref{MSE1}:
\begin{equation} \label{MSE2}
    \mathbb{E}[(\hat{\tau}_k - \tau_k)^2] = \mathbb{E}\left[ \left( \frac{\dot{f}(\tau_k)}{\ddot{f}(\tau_k)} \right)^2 \right]
    \approx \frac{\mathbb{E}[|\dot{f}(\tau_k)|^2]}{|\mathbb{E}[\ddot{f}(\tau_k)]|^2}.
\end{equation}
Such approximations are commonly used in the analysis of nonlinear estimators, and are widely accepted when the curvature of the objective function is stable and the estimator remains in a local neighborhood of the true parameter. 

To relate the MSE expression \eqref{MSE2} to $s_{\text{MF}}(\tau)$, we have the first derivative of $f(\tau)$ at the true delay $\tau_k$ as
\begin{align}
        \dot{f}(\tau_k)=2 \Re\{\dot{s}_{\text{MF}}(\tau_k) s^*(\tau_k)\} \approx 2N|\alpha_k|\Re\{\dot{s}_{\text{MF}}(\tau_k)\},
\end{align}
where $s^*(\tau_k) \approx N\alpha_k^*$ is the conjugate of the target peak amplitude. Thus, the expectation of its square can be rewritten as
\begin{align}
    \mathbb{E}[|\dot{f}(\tau_k)|^2] &= 4N^2|\alpha_k|^2\mathbb{E}[(\Re\{\dot{s}_{\text{MF}}(\tau_k)\})^2] \nonumber \\
    &=2N^2|\alpha_k|^2\mathbb{E}[|\dot{s}_{\text{MF}}(\tau_k)|^2].
\end{align}
Similarly, the second derivative is given by
\begin{align}
    \ddot{f}(\tau_k) &=2 \Re\{\ddot{s}_{\text{MF}}(\tau_k) s^*(\tau_k) + \dot{s}_{\text{MF}}(\tau_k) \dot{s}^*(\tau_k)\} \nonumber \\
    & \approx  2N|\alpha_k|\Re\{\ddot{s}_{\text{MF}}(\tau_k)\}.
\end{align}
Thus, the square of its expectation becomes 
\begin{align}
    |\mathbb{E}[\ddot{f}(\tau_k)]|^2 &= 4N^2|\alpha_k|^2|\mathbb{E}[\ddot{s}_{\text{MF}}(\tau_k)]|^2.
\end{align}
Substituting those into \eqref{MSE2} yields
\begin{equation}
        \mathbb{E}[(\hat{\tau}_k - \tau_k)^2] = \frac{\mathbb{E}[|\dot{s}_{\text{MF}}(\tau_k)|^2]}{2|\mathbb{E}[\ddot{s}_{\text{MF}}(\tau_k)]|^2}. \label{MSE3}
\end{equation}

Now, we are ready to derive a closed-form expression of the delay estimation MSE with MF receiver. The first derivative of $s_{\text{MF}}(\tau)$ at $\tau_k$ with zero-mean is written as 
\begin{align} \label{firstDev1}
    \dot{s}_{\text{MF}}(\tau_k)  =  (j2\pi \Delta f) & \underbrace{\left(\sum_{j \ne k}^K \sum_{n=0}^{N-1} n \alpha_j |x_n|^2 e^{j2\pi n \Delta f (\tau_k - \tau_j)}  \right.}_{\text{Sidelobe interference from other targets}, \;\mathrm{C}_{1}}   \nonumber \\ 
     & + \underbrace{\left. \sum_{n=0}^{N-1} n x_n^* z_n e^{j2\pi n \Delta f \tau_k}\right)}_{\text{Noise}, \;\mathrm{C}_{2}},
\end{align}
where we can observe that the first derivative of $s_{\text{MF}}(\tau)$ depends on the sidelobe of other targets and receiver noises. Again, the second derivative $s_{\text{MF}}(\tau)$ at $\tau_k$ is given by
\begin{align} \label{SecondDev1}
    \ddot{s}_{\text{MF}}(\tau_k) = (j2\pi\Delta f)^{2} & \left(
    \sum_{j=1}^K \sum_{n=0}^{N-1} n^2\,\alpha_j\,|x_n|^2\,e^{j2\pi n\Delta f(\tau_k-\tau_j)}\;\right. \nonumber \\
    & \left.+\;\sum_{n=0}^{N-1} n^2\,x_n^*\,z_n\,e^{j2\pi n\Delta f\tau_k}\right).
\end{align}
With \eqref{firstDev1} and \eqref{SecondDev1} on hands, we get the following theorem for delay estimation MSE under MF receiver.
\begin{theorem}\label{theo2}
A closed-form expression for the delay estimation MSE for target $k$ with the MF receiver and the conventional delay estimator is given by
\begin{align}
    \mathrm{MSE}_{\text{MF},k} = \frac{3\left((\mu_{4}-1)\sum_{j\neq k}^K|\alpha_{j}|^2 + \sigma^2\right)}
     {8\pi^2\Delta f^2 |\alpha_{k}|^2 N^3}. \label{Theo_eq2}
\end{align}
\end{theorem}
\renewcommand\qedsymbol{$\blacksquare$}
\begin{proof}
Please refer to Appendix \ref{proof_theo2}.
\end{proof}

It is observed that the delay estimation MSE with the classic MF receiver is affected by the sidelobes of other delay sources, whose magnitudes depend on the constellation and its fourth moment. This performance trend is consistent with the ambiguity function analysis in \cite{liu2025cp} and the effective signal-to-noise-plus-interference ratio (SINR) analysis in \cite{geiger2025joint}. Moreover, it implies that the MF receiver with conventional estimators such as FFT-based estimator or subspace-based MUSIC may not achieve the CRB, even under high SNR and a sufficiently large number of samples, when $\mu_4 > 1$ and $K > 1$. It is worth noting that these receivers remain efficient when a unit-amplitude constellation with $\mu_4=1$ is employed, or when only a single target exists ($K=1$).

\subsection{Delay Estimation MSE with Reciprocal Filtering Receiver}
We investigate a mismatched filtering receiver, which suppresses the sidelobes of ambiguity function (AF) by sacrificing SNR. Particularly, we focus on RF known for equalizing data sequences in random ISAC signaling \cite{wojaczek2018reciprocal}. The RF receiver is implemented by element-wise division of $\mathbf{y}$ by $\mathbf{x}$, of which output $\mathbf{y}_{\text{RF}} = \mathbf{y} \oslash \mathbf{x}$ is represented as
\begin{align}
    \mathbf{y}_{\text{RF}} = \mathbf{a}^T \mathbf{H}  + \mathbf{z}_{\text{RF}}.
\end{align}
It is observed that the output of RF has no effect of random signaling in the signal term, implying that it is free from the sidelobes of other delay sources. Notably, $\mathbf{z}_{\text{RF}}$ is zero-mean, but its variance is reshaped due to RF, which is given by \cite{han2025secure}
\begin{align}
    \mathrm{Var}(z_{\text{RF},n}) = \mathbb{E}\!\left[\left|z_n x_n^{-1}\right|^2\right]
    = \sigma^2 \,\mathbb{E}\bigl[|x_n|^{-2}\bigr] = \sigma^2\,\nu_{-2}.
\end{align}
Similar to MF receiver, the delay estimator in RF receiver is defined as
\begin{align} \label{estimator}
    \hat{\tau} = \arg \max_{\tau \in \mathcal{T}} \left| \mathbf{h}^H(\tau) \mathbf{y}_{\text{RF}}^T \right|, 
\end{align}
where $\mathbf{h}(\tau) \in \mathcal{A}$. By defining $s_{\text{RF}}(\tau) = \mathbf{h}^H(\tau) \mathbf{y}_{\text{RF}}^T$, we have
\begin{equation}
    s_{\text{RF}}(\tau) = \sum_{k=1}^K \sum_{n=0}^{N-1} \alpha_k e^{j2\pi n \Delta f (\tau - \tau_k)} + \sum_{n=0}^{N-1} z_{\text{RF},n} e^{j2\pi n \Delta f \tau}.
\end{equation}
Based on Lemma \ref{lemma1}, the following theorem provides a closed-form expression for the delay estimation MSE with RF receiver.
\begin{theorem}\label{theo3}
A closed-form expression for the delay estimation MSE for target $k$ with the RF receiver and the conventional delay estimator is given by
\begin{align}
    \mathrm{MSE}_{\text{RF},k} = \frac{3 \sigma^2 \nu_{-2}}
     {8\pi^2\Delta f^2 |\alpha_{k}|^2 N^3}. \label{Theo_eq3}
\end{align}
\end{theorem}
\renewcommand\qedsymbol{$\blacksquare$}
\begin{proof}
Please refer to Appendix \ref{proof_theo3}.
\end{proof}
Unlike the delay estimation MSE of the MF receiver, the MSE of the RF receiver remains unaffected by the presence of other sources. Instead, it suffers an SNR loss proportional to the inverse second moment of the modulation constellation, implying that its MSE is always larger than the CRB unless a unit-amplitude constellation is employed. Therefore, the RF receiver also cannot achieve the lower bound performance, even at high SNR, when using modulation constellations with $\nu_{-2} > 1$.

\begin{comment}
Building on the analysis of MSE for MF- and RF-based receivers with conventional estimators, we state the following corollary regarding the relationship between the MSEs and the CRB.  
\begin{corollary}\label{corr1}
    For $K > 1$, the MSE of delay estimation with MF and RF receivers based on conventional estimators using atoms in the set $\mathcal{A}=\{\mathbf{h}(\tau) \; | \; \tau \in \mathcal{T}\}$ asymptotically achieves the expected CRB if and only if the communication data are modulated by a rotation-symmetric unit-amplitude constellation, such as PSK.
\end{corollary}
\begin{proof}
For sufficiently high SNR, we have $\mathrm{MSE}_{\text{MF}} = \mathrm{MSE}_{\text{RF}} = \mathbb{E}[\mathrm{CRB}]$ if and only if $\mu_4 = \nu_{-2} = 1$. It is straightforward to show that the condition $\mu_4 = \nu_{-2} = 1$ is satisfied only by unit-amplitude constellations.
\end{proof}
\end{comment}

\textbf{Remark 1:} Throughout the presented analysis, we observe that the MSEs of MF- and RF-based receivers are differently influenced by the modulation constellation and multi-target interference. For MF, the delay estimation performance depends on the fourth moment $\mu_4$ of the constellation and on the interference from other targets, quantified by $\sum_{j \neq k}^K |\alpha_{j}|^2$. This indicates that the sidelobes of other targets act as interference to the desired target, leading to degraded estimation performance. In contrast, for RF, the performance depends only on the inverse second moment $\nu_{-2}$ of the constellation, while the sidelobes induced by other targets are completely removed. Therefore, the performance comparison between MF and RF receivers must consider both the choice of modulation constellation and the number of targets or clutter sources present in the ISAC scenario.

\subsection{ISAC Constellation Design}
Building on the MSE analysis of OFDM-ISAC systems under different receiver structures, we propose a geometric constellation shaping (GCS) framework that provides a flexible trade-off between sensing and communication performance. To this end, we adopt the minimum Euclidean distance (MED) as the communication performance metric, motivated by the fact that the bit-error rate is dominated by errors occurring at the MED \cite{caire2002bit}. Accordingly, the ISAC constellation design problem is formulated as
\begin{subequations}\label{Eqn::P1}
    \begin{align}
    & \underset{\{s_m\}_{m=1}^M}{\text{minimize}}
    \qquad \rho f_s(\mathcal{S}) - (1-\rho) d_{\min} \\
    & \text{subject to} \qquad |s_i - s_j| \geq d_{\min}, \; \forall s_i \neq s_j \in \mathcal{S}, \\
    & \qquad \qquad \quad \mathbb{E}[s_i] = 0, \; \mathbb{E}[s_i^2] = 0, \; \mathbb{E}[|s_i|^2] = 1,
    \end{align}
\end{subequations}
where $\rho \in [0,1]$ denotes the priority weight between sensing and communication and $d_{\min}$ is the MED between symbols. The objective function involves the sensing performance metric $f_s(\mathcal{S})$, which depends on the receiver structure. Specifically, for the MF receiver, $f_s(\mathcal{S}) = \mu_4$ to minimize delay-estimation MSE, while for the RF receiver, $f_s(\mathcal{S}) = \nu_{-2}$ to reduce SNR loss, thereby minimizing the MSE. Since the above formulation leads to a non-linear optimization problem, we employ sequential quadratic programming (SQP) with multiple random initializations to obtain constellation solutions.

\section{Simulation and Experimental Validation}
In this section, we validate the proposed theory and closed-form expressions through numerical simulations and a proof-of-concept (PoC) demonstration. To examine the impact of the modulation constellation on sensing performance, we consider various constellations: QPSK, 16QAM, 64QAM, and a customized 32APSK with four rings, whose radius ratio is $1:2:3:4$. 

\subsection{Numerical Simulation Results}

\begin{figure}[t!]
    \centering
    \subfigure[]{\includegraphics[width=0.24\textwidth]{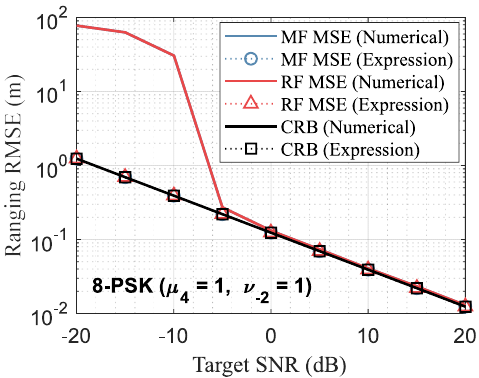}}
    \subfigure[]{\includegraphics[width=0.24\textwidth]{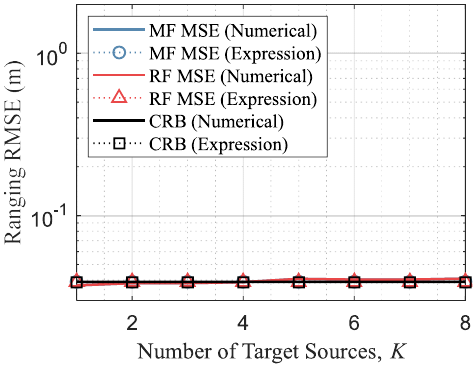}}\\
    \subfigure[]{\includegraphics[width=0.24\textwidth]{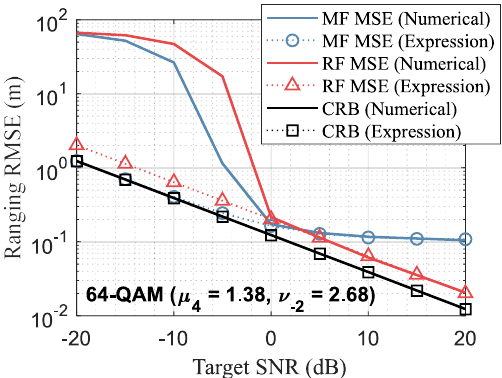}}
    \subfigure[]{\includegraphics[width=0.24\textwidth]{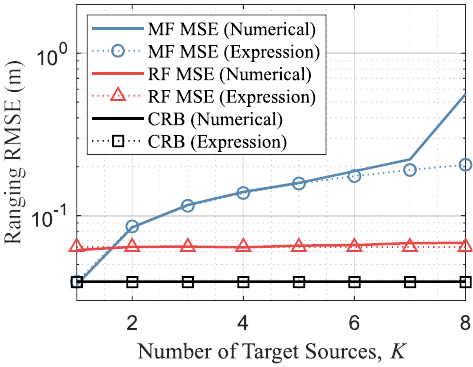}} \\
    \subfigure[]{\includegraphics[width=0.24\textwidth]{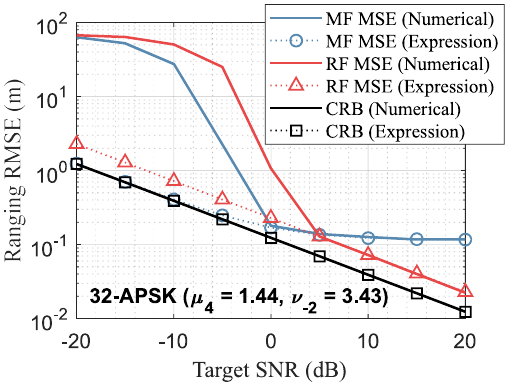}}
    \subfigure[]{\includegraphics[width=0.24\textwidth]{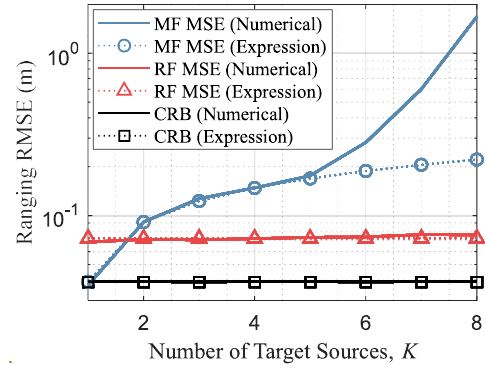}}
    \caption{Range estimation accuracy under three different modulation constellations. Subplots (a), (c), and (e) show the range estimation root-MSE (RMSE) versus target SNR with $K = 3$, while (b), (d), and (f) present the RMSE versus the number of targets at 10~dB SNR for QPSK, 64QAM, and 32APSK, respectively.}
    \label{Fig::1}
\end{figure}

For the numerical simulations, we adopt a CP-OFDM waveform with $N = 256$, $B = 50$~MHz, and a CP length of $0.64~\mu$s. Each result is averaged over 1000 Monte-Carlo runs. The target SNR is defined as $|\alpha_k|^2/\sigma^2$. A subspace-based matrix pencil estimator is employed for range estimation after MF and RF processing.  

First, we evaluate the range estimation accuracy under different modulation constellations to verify the presented theory and closed-form expressions with numerical results. As shown in Fig.~\ref{Fig::1}, the CRB remains consistent regardless of the modulation constellation, while the RMSE of MF- and RF-based receivers degrades with higher values of $\mu_4$ and $\nu_{-2}$, respectively. The derived closed-form expressions for the MSE closely match the numerical estimation results when the target SNR is larger than 0~dB. Furthermore, it is observed that the MSE of the MF receiver saturates as the SNR increases, since sidelobe interference from other targets limits the range estimation performance, whereas the RF-based receiver is free from such sidelobe effects.

As shown in Fig.~\ref{Fig::1}(b), (d), and (f), the impact of the number of targets $K$ on the range estimation performance is illustrated for both MF and RF receivers. For QPSK, which yields the optimal sensing performance, the accuracy remains constant regardless of $K$. In contrast, as the number of targets increases, the performance of MF degrades due to the growth of sidelobe interference induced by random ISAC signaling with $\mu_4 > 1$. On the other hand, the performance of RF remains constant, but it exhibits a gap from the CRB owing to the SNR loss when $\nu_{-2} > 1$. Overall, the presented closed-form expressions for the MSEs closely match the numerical evaluations, providing a tractable characterization of estimation performance for sensing with OFDM communication data payloads.

\subsection{Experimental Results}
\begin{figure}[t!]
    \centering
    {\includegraphics[width=0.48\textwidth]{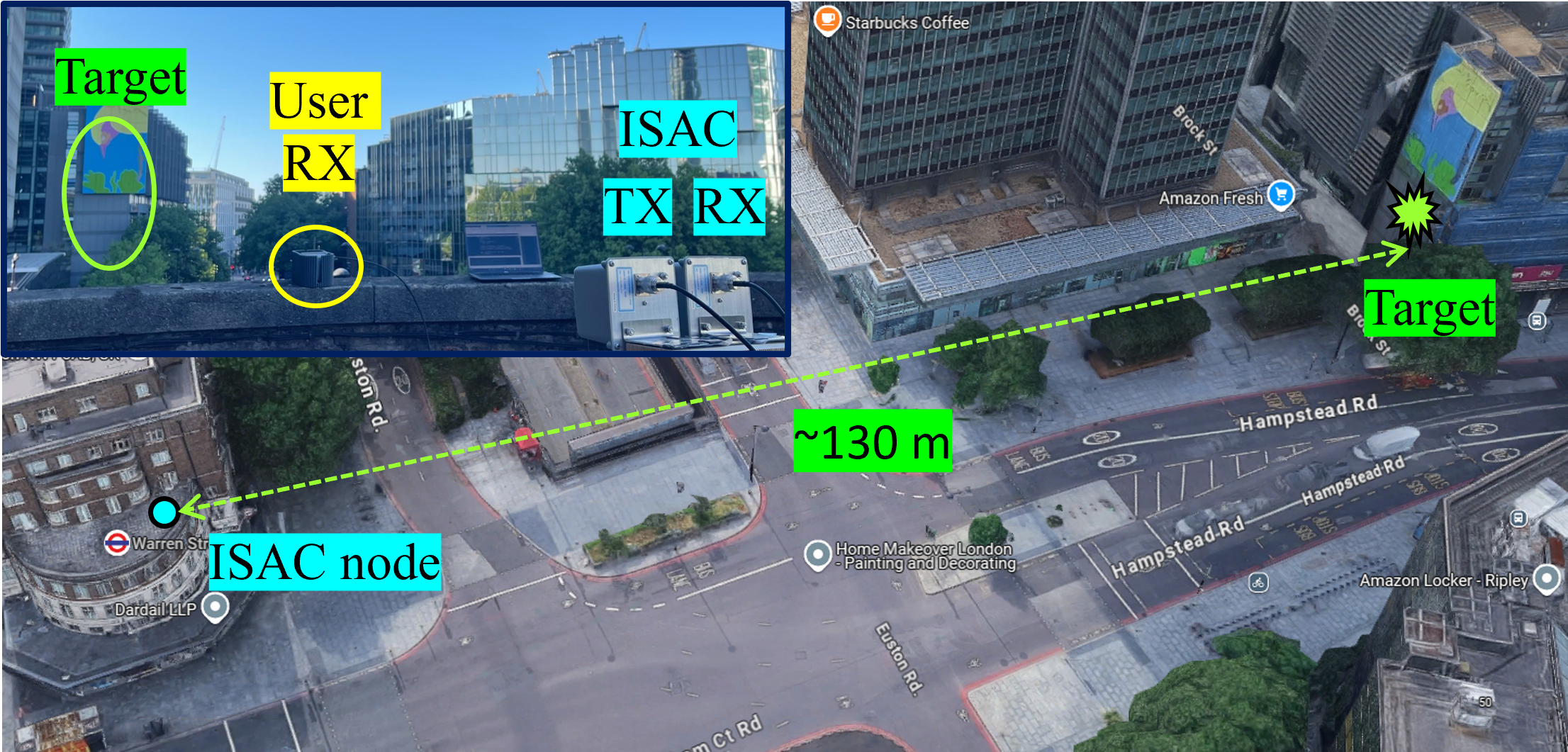}}
    \caption{Photograph of the measurement setup and corresponding Google Map of the London area \cite{Googlemap}.}
    \label{Fig::2}
\end{figure}

\begin{figure}[t!]
    \centering
    \subfigure[]{\includegraphics[width=0.45\textwidth]{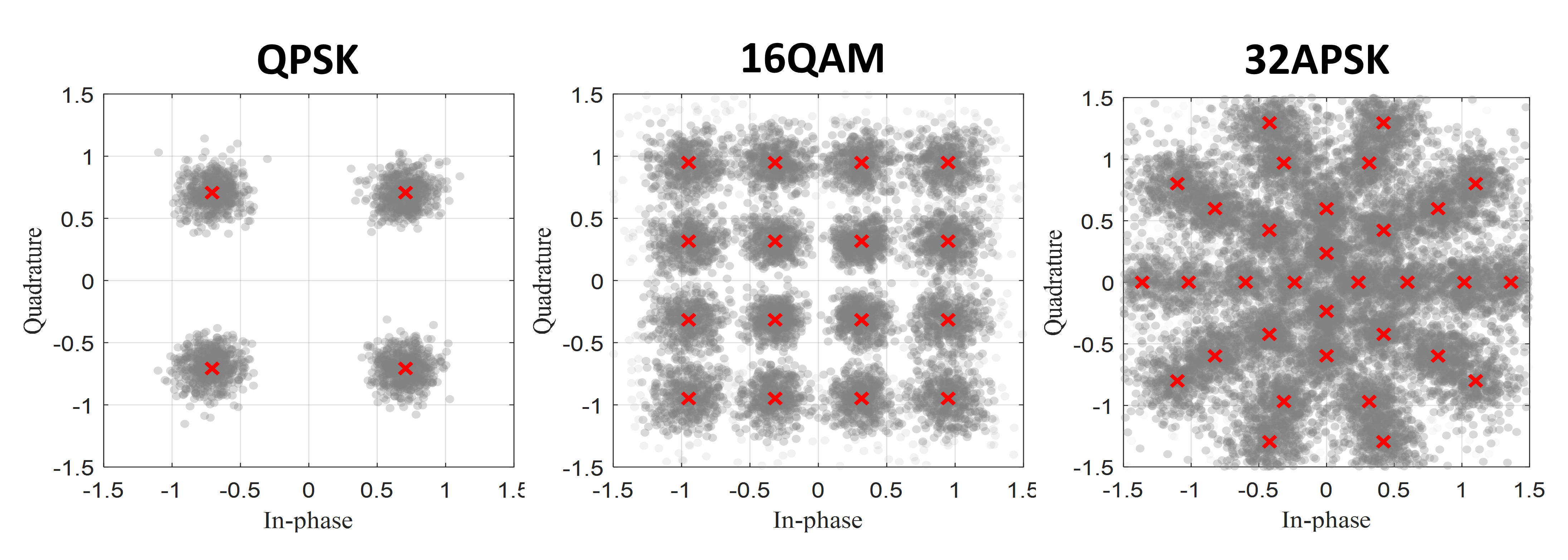}} \\
    \subfigure[]{\includegraphics[width=0.4\textwidth]{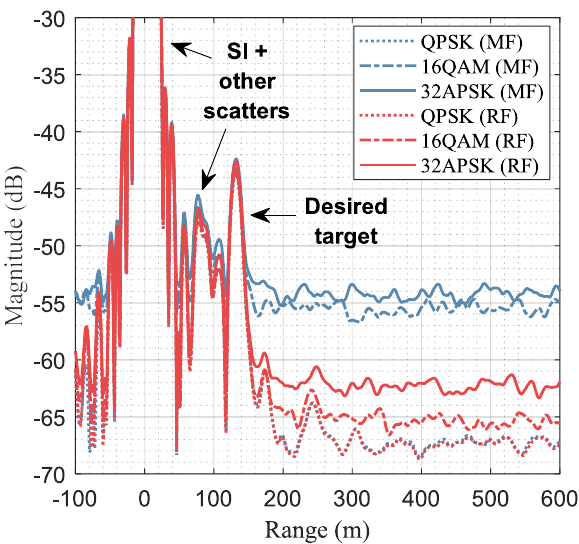}}
    \caption{Experimental results of OFDM-ISAC: (a) received symbols at the communication user and (b) range profiles at the sensing receiver.}
    \label{Fig::3}
\end{figure}

\begin{figure}[t!]
    \centering
    \subfigure[]{\includegraphics[width=0.5\textwidth]{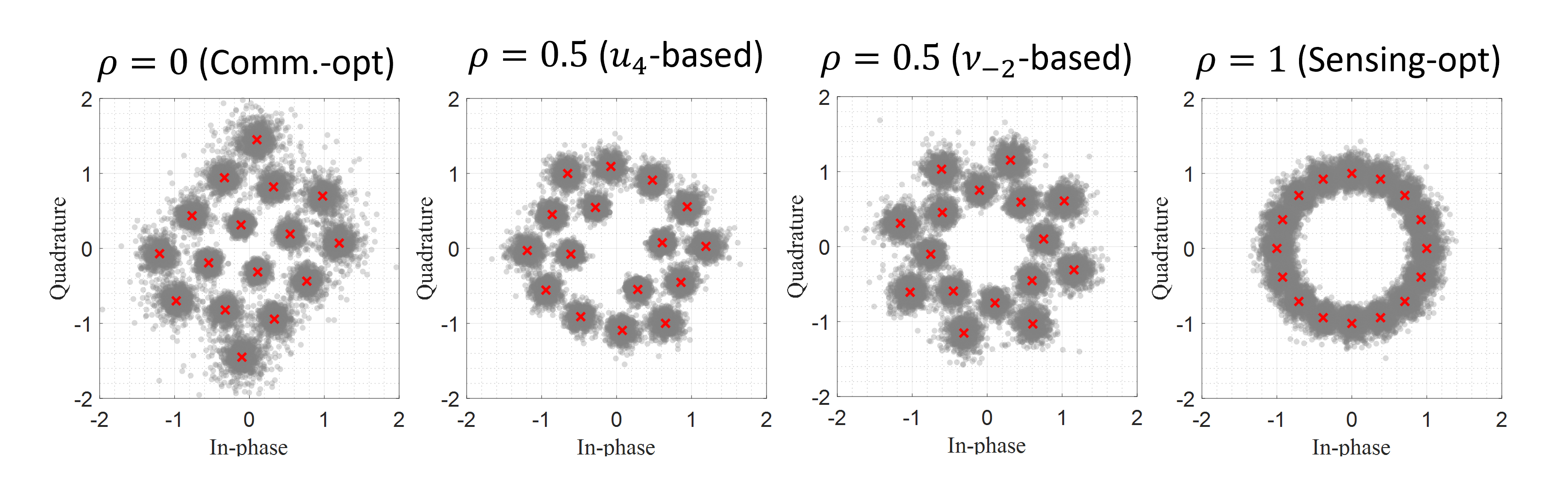}} \\
    \subfigure[]{\includegraphics[width=0.23\textwidth]{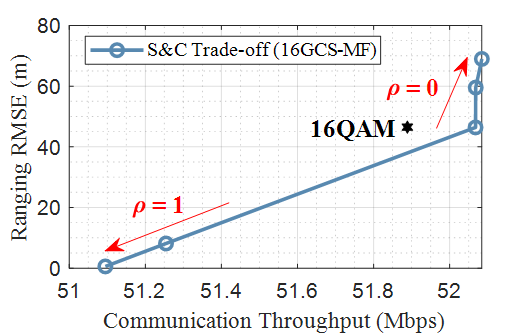}}
    \subfigure[]{\includegraphics[width=0.23\textwidth]{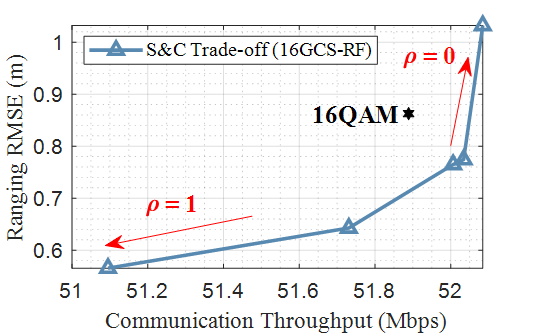}}
    \caption{Experimental results of the proposed constellation design for OFDM-ISAC: (a) received symbols at the communication user and ranging RMSE vs. communication throughput trade-offs under (b) MF and (c) RF sensing receiver.}
    \label{Fig::4}
\end{figure}

Beyond numerical simulations, we further validate the proposed analysis through over-the-air experiments. To this end, a software-defined radio (AD9363) with directional antennas of 12~dBi gain is employed. The desired target, a building wall with high reflectivity, is located approximately 130~m from the ISAC node. The transmitted OFDM-ISAC signal is simultaneously received at the user receiver, as illustrated in Fig.~\ref{Fig::2}. The center frequency is set to 2.4~GHz with a bandwidth of $B = 20$~MHz and $N = 512$ subcarriers. In addition, 512 OFDM symbols are transmitted per frame, providing a coherent processing gain of 27~dB for radar sensing. Each measurement result is obtained by averaging over 100 independent frames.

Fig.~\ref{Fig::3}(a) shows the received data symbols at the user for three different constellations (QPSK, 16QAM, and 32APSK), while Fig.~\ref{Fig::3}(b) illustrates the range profiles of MF and RF. Notably, the effective noise floor of RF is lower than that of MF except in the case of QPSK modulation. This is because strong self-interference between TX and RX, together with unexpected clutter sources at near ranges, generates high sidelobes in the MF output, thereby increasing the effective noise floor. In contrast, that of RF increases as the inverse second moment $\nu_{-2}$ of the constellation increases. It is worth noting that the RF receiver may outperform in the practical scenario with multiple scatterers in the case that the data payloads are modulated by the signal constellation with $\mu_4 > 1, \nu_{-2}>1$. This is because the effect of sidelobe interference increases as much as the number of scatterers, which greatly degrade the sensing performance as seen in Fig.~\ref{Fig::3}(b).

Finally, we show the experimental results with the designed constellation with $M=16$. Fig.~\ref{Fig::4}(a) shows the received signal constellation at user, which is designed upon varying priority weight $\rho$. Notably, $\mu_4$-based design for MF and $\nu_{-2}$-based design for RF yield different constellation geometry, implying that the ISAC constellation design has to consider which sensing receiver may be employed. Moreover, the experimental trade-off between sensing and communication is provided in Fig.~\ref{Fig::4}(b) and (c) for MF and RF receivers, respectively. Compared to conventional 16QAM constellation, the designed ISAC constellation enhances the ISAC performance with flexible control between sensing and communication functional priorities. 

\section{Conclusion}
In this paper, we have presented an analytical framework for OFDM-based ISAC. An estimation-theoretic approach was developed to characterize the range estimation performance under random ISAC signaling with two specific receiver processing schemes: MF and RF. Notably, the impact of the modulation constellation and the presence of multiple targets differ depending on the receiver processing, as demonstrated through theoretical analysis, numerical simulations, and PoC experiments. Moreover, the ISAC constellation design is presented and verified through experiments, showing range estimation RMSE vs. communication throughput trade-off under the specific sensing receiver. The proposed framework offers new insights into the design of OFDM-ISAC transmitters and receivers, opening new opportunity in ISAC transceiver design.

\appendices
\section{Proof of Theorem \ref{theo2}}\label{proof_theo2}
Let us define the first term related to sidelobe effects as $\mathrm{C}_{1}$ and the second term as $\mathrm{C}_{2}$ in \eqref{firstDev1}. As cross-products of $\mathrm{C}_{1}$ and $\mathrm{C}_{2}$ are zero-mean, we only need to consider the expectation of the squared magnitude of each term. The expectation of the squared magnitude of $\mathrm{C}_{1}$ is given by

\begin{align}\label{SLMSE}
    \mathbb{E}\left[|\mathrm{C}_{1}|^2\right] & = \sum_{j \ne k}^K |\alpha_j|^2 \left(\mu_4 \sum_{n=0}^{N-1} n^2 + \left| \sum_{n=0}^{N-1} n e^{j\phi_{j,n}} \right|^2 - \sum_{n=0}^{N-1} n^2
    \right) \notag \\
    & \approx (\mu_4 - 1)\sum_{j \ne k}^K |\alpha_j|^2  \sum_{n=0}^{N-1} n^2,
\end{align}
where \( \phi_{j,n} = 2\pi n \Delta f (\tau_k - \tau_j) \). Here, it should be noted that the cross-terms between different target delays $\left| \sum_{n=0}^{N-1} n e^{j\phi_{j,n}} \right|^2, \; \forall{j \neq k},$ are negligible as they are much smaller than $\sum_{n=0}^{N-1} n^2$ unless any two targets are closely located.
Also, the expectation of the squared magnitude of $\mathrm{C}_{2}$ is expressed as
\begin{align}\label{NoiseMSE}
    \mathbb{E}[|\mathrm{C}_{2}|^2] 
    &= \sum_{n=0}^{N-1} n^2 \, \mathbb{E}[|x_n^* z_n|^2] 
    = \sigma^2 \sum_{n=0}^{N-1} n^2.
\end{align}
Combining \eqref{SLMSE} and \eqref{NoiseMSE} together, we get
\begin{align} \label{firstexp}
    \mathbb{E}[|\dot{s}_{\text{MF}}(\tau_k)|^2] = (2\pi \Delta f)^2 \sum_{n=0}^{N-1} n^2 \left((\mu_4 - 1)\sum_{j \ne k}^K |\alpha_j|^2  + \sigma^2\right).
\end{align}

Similarly, the expectation of the second derivative at $\tau_k$ is obtained as
\begin{align} \label{secondexp}
    \mathbb{E}[\ddot{s}_{\text{MF}}(\tau_k)] = (j2\pi\Delta f)^{2}\,\alpha_k \sum_{n=0}^{N-1}n^2.
\end{align}
This is because the sidelobe interference can be ignored with sufficient delay difference $|\tau_k - \tau_j| \geq 1/B$ and the noise term has zero-mean. 
Plugging \eqref{firstexp} and \eqref{secondexp} into \eqref{MSE3} yields \eqref{Theo_eq2}, completing the proof.

\section{Proof of Theorem \ref{theo3}}\label{proof_theo3}
We derive the expectation of the squared magnitude of the first derivative and that of the second derivative respectively. 
\begin{align}
    \mathbb{E}\bigl[\lvert \dot{s}_{\rm RF}(\tau_k)\rvert^2\bigr]
    &= (2\pi\Delta f)^2\,\sum_{n=0}^{N-1}n^2\,\sigma^2\nu_{-2}, \\
    \left(\mathbb{E}\bigl[\ddot{s}_{\rm RF}(\tau_k)\bigr]\right)^2
    &= (j2\pi\Delta f)^4\,|\alpha_k|^2\,\left(\sum_{n=0}^{N-1}n^2\right)^2.
\end{align}
Plugging those into \eqref{MSE3} gives us
\begin{align}
    \frac{\mathbb{E}\bigl[\lvert \dot{s}_{\rm RF}(\tau_k)\rvert^2\bigr]}
         {2\bigl\lvert\mathbb{E}[\ddot{s}_{\rm RF}(\tau_k)]\bigr\rvert^2}
    = \frac{\sigma^2\,\nu_{-2}}
           {2\,\lvert\alpha_k\rvert^2\,(2\pi\Delta f)^2\,\sum_{n=0}^{N-1}n^2}. 
\end{align}
For sufficiently large $N$, it approximates to \eqref{Theo_eq3}, completing the proof.

\bibliographystyle{IEEEtran}
\bibliography{IEEEabrv,reference}

\begin{thebibliography}{10}
\providecommand{\url}[1]{#1}
\csname url@samestyle\endcsname
\providecommand{\newblock}{\relax}
\providecommand{\bibinfo}[2]{#2}
\providecommand{\BIBentrySTDinterwordspacing}{\spaceskip=0pt\relax}
\providecommand{\BIBentryALTinterwordstretchfactor}{4}
\providecommand{\BIBentryALTinterwordspacing}{\spaceskip=\fontdimen2\font plus
\BIBentryALTinterwordstretchfactor\fontdimen3\font minus \fontdimen4\font\relax}
\providecommand{\BIBforeignlanguage}[2]{{%
\expandafter\ifx\csname l@#1\endcsname\relax
\typeout{** WARNING: IEEEtran.bst: No hyphenation pattern has been}%
\typeout{** loaded for the language `#1'. Using the pattern for}%
\typeout{** the default language instead.}%
\else
\language=\csname l@#1\endcsname
\fi
#2}}
\providecommand{\BIBdecl}{\relax}
\BIBdecl

\bibitem{luo2025isac}
X.~Luo, Q.~Lin, R.~Zhang, H.-H. Chen, X.~Wang, and M.~Huang, ``{ISAC--A Survey on Its Layered Architecture, Technologies, Standardizations, Prototypes and Testbeds},'' \emph{IEEE Communications Surveys \& Tutorials}, 2025.

\bibitem{prasad2004ofdm}
R.~Prasad, \emph{{OFDM for wireless communications systems}}.\hskip 1em plus 0.5em minus 0.4em\relax Artech House, 2004, vol.~2.

\bibitem{liu2025cp}
F.~Liu, Y.~Zhang, Y.~Xiong, S.~Li, W.~Yuan, F.~Gao, S.~Jin, and G.~Caire, ``{CP-OFDM achieves the lowest average ranging sidelobe under QAM/PSK constellations},'' \emph{IEEE Transactions on Information Theory}, 2025.

\bibitem{knill2021coded}
C.~Knill, F.~Embacher, B.~Schweizer, S.~Stephany, and C.~Waldschmidt, ``{Coded OFDM waveforms for MIMO radars},'' \emph{IEEE Transactions on Vehicular Technology}, vol.~70, no.~9, pp. 8769--8780, 2021.

\bibitem{liu2025uncovering}
F.~Liu, Y.~Xiong, S.~Lu, S.~Li, W.~Yuan, C.~Masouros, S.~Jin, and G.~Caire, ``{Uncovering the iceberg in the sea: Fundamentals of pulse shaping and modulation design for random ISAC signals},'' \emph{IEEE Transactions on Signal Processing}, 2025.

\bibitem{geiger2025joint}
B.~Geiger, F.~Liu, S.~Lu, A.~Rode, and L.~Schmalen, ``{Joint optimization of geometric and probabilistic constellation shaping for OFDM-ISAC systems},'' in \emph{2025 IEEE 5th International Symposium on Joint Communications \& Sensing (JC\&S)}.\hskip 1em plus 0.5em minus 0.4em\relax IEEE, 2025, pp. 1--6.

\bibitem{mcaulay1971optimal}
R.~McAulay and J.~Johnson, ``{Optimal mismatched filter design for radar ranging, detection, and resolution},'' \emph{IEEE Transactions on Information Theory}, vol.~17, no.~6, pp. 696--701, 1971.

\bibitem{sturm2011waveform}
C.~Sturm and W.~Wiesbeck, ``{Waveform design and signal processing aspects for fusion of wireless communications and radar sensing},'' \emph{Proceedings of the IEEE}, vol.~99, no.~7, pp. 1236--1259, 2011.

\bibitem{wojaczek2018reciprocal}
P.~Wojaczek, F.~Colone, D.~Cristallini, and P.~Lombardo, ``{Reciprocal-filter-based STAP for passive radar on moving platforms},'' \emph{IEEE Transactions on Aerospace and Electronic Systems}, vol.~55, no.~2, pp. 967--988, 2018.

\bibitem{keskin2025fundamental}
M.~F. Keskin, M.~M. Mojahedian, J.~O. Lacruz, C.~Marcus, O.~Eriksson, A.~Giorgetti, J.~Widmer, and H.~Wymeersch, ``{Fundamental trade-offs in monostatic ISAC: A holistic investigation towards 6G},'' \emph{IEEE Transactions on Wireless Communications}, 2025.

\bibitem{han2025secure}
K.~Han, K.~Meng, and C.~Masouros, ``{Secure-Sensing ISAC: Ambiguity Function Engineering for Imparing Unauthorized Sensing},'' \emph{IEEE Transactions on Wireless Communications}, 2025.

\bibitem{caire2002bit}
G.~Caire, G.~Taricco, and E.~Biglieri, ``{Bit-interleaved coded modulation},'' \emph{IEEE transactions on information theory}, vol.~44, no.~3, pp. 927--946, 2002.

\bibitem{Googlemap}
\BIBentryALTinterwordspacing
Google. (2025) Euston road, london, uk. [Online]. Available: \url{https://www.google.com/maps/@51.524691,-0.1366646,76a,35y,288.62h,65.66t}
\BIBentrySTDinterwordspacing

\end{thebibliography}
\end{document}